\documentclass[prodmode,acmjetc]{acmsmall} 

\usepackage{epsfig}
\usepackage{graphicx}
\usepackage{color}
\usepackage{graphics}
\usepackage{dcolumn}
\usepackage{bm}
\usepackage{multirow}
\usepackage{longtable}
\usepackage{amsmath}

\input{Qcircuit.tex}

\begin{document}

\markboth{M. Arabzadeh et al.}{Quantum-Logic Synthesis of Hermitian Gates}

\title{Quantum-Logic Synthesis of Hermitian Gates}
\author{MONA ARABZADEH
\affil{Amirkabir University of Technology}
MAHBOOBEH HOUSHMAND
\affil{Amirkabir University of Technology}
MEHDI SEDIGHI
\affil{Amirkabir University of Technology}
MORTEZA SAHEB~ZAMANI
\affil{Amirkabir University of Technology}}

\begin{abstract}
In this paper, the problem of synthesizing a general Hermitian quantum gate into a set of primary quantum gates is addressed.
To this end, an extended version of the Jacobi approach for calculating the eigenvalues of Hermitian matrices in linear algebra is considered as the basis of the
proposed synthesis method. The quantum circuit synthesis method derived from the Jacobi approach and its optimization challenges are described. It is shown that the proposed method results in multiple-control rotation gates around the $y$ axis, multiple-control phase shift gates, multiple-control NOT gates and a middle diagonal Hermitian matrix, which can be synthesized to multiple-control Pauli \emph{Z} gates. 
Using the proposed approach, it is shown how multiple-control $U$ gates, where $U$ is a single-qubit Hermitian quantum gate, can be implemented using a linear number of elementary gates in terms of circuit lines with the aid of one auxiliary qubit in an arbitrary state.
\end{abstract}

%
%


%
%




%

\maketitle

\section{Introduction}\label{sec:level1}
Quantum computers use quantum mechanical phenomena, such as superposition
and entanglement  which make them advantageous over the classical ones, e.g., they can solve
certain problems such as integer factorization~\cite{shor-1997-26} and database search~\cite{Grover} much more quickly than any classical computer using the best currently known algorithms.
Besides that, the idea of simulating quantum-mechanical effects by computers \cite{Feynman82} can be further considered as a motivation for working on quantum computation problems.

Quantum-logic synthesis is referred to as the problem of decomposing a given arbitrary quantum function to a set of quantum gates, i.e., elementary operations
which can be implemented in quantum technologies.
This problem has been widely considered for a general unitary matrix and a number of solutions based on matrix decomposition have been proposed.
QR decomposition in \cite{George2001,vartiainen-2004} and the cosine-sine (CS) decomposition in CSD \cite{Bergholm:2004}, QSD \cite{Shende06} and their combination, BQD \cite{saeedi2010block}, have been applied for quantum-logic synthesis. On the other hand, a set of methods for more specific unitary matrices such as two-qubit operators \cite{vidal-2004-69,Shende-1-2004,vatan-2004-69}, three-qubit operators~\cite{vatan-2004}, general diagonal matrices \cite{bullock-2004-4} and diagonal Hermitian quantum gates~\cite{diagonal} were proposed in the literature. It is shown in~\cite{diagonal} that diagonal Hermitian quantum gates can be decomposed to a set that solely consists of multiple-controlled \emph{Z} gates.

Since quantum gates are linear unitary transformations, they are invertible. The inverse of a unitary matrix $U$ is  $U^{-1}=U^\dag$, which is its conjugate transpose. A restricted set of unitary operations are those which are self-inverse, i.e., $U=U^{-1}=U^\dag$. These sets are called Hermitian quantum gates.
Many quantum gates such as CNOT, SWAP, Toffoli, Fredkin, Hadamard, and Pauli gates, which are used frequently in quantum circuits, are Hermitian~\cite{pathak13}.
It is notable that if a gate is Hermitian, then its controlled gate with any number of control lines is also Hermitian.
Hermitian gates also appear in well-known quantum circuits,
such as encoding and decoding circuits of stabilizer codes~\cite{Grassal,stab} where these circuits solely consist of Hadamard, Pauli and controlled-Pauli gates. 
Besides, Hermitian matrices, according to their properties, are used as a specific set of matrices in quantum algorithms. In \cite{harrow2009quantum}, the solution to a set of equations is obtained by a quantum algorithm assuming that the matrix of coefficient is a sparse Hermitian matrix.
Subsequently, another research team~\cite{wiebe2012quantum} applied the same algorithm to determine the quality of a least-squares fit over an exponentially large data set.

In this paper, the decomposition problem of Hermitian unitary matrices using an extended version of the Jacobi idea~\cite{golub2012matrix} is addressed.
The Jacobi method in linear algebra is an iterative method to find the eigenvalues and eigenvectors of symmetric and Hermitian matrices.
Using this, a given Hermitian quantum gate is synthesized to a set of multiple-control gates and a diagonal Hermitian matrix by the proposed approach.
The decomposition of these high-level gates to CNOT and single-qubit gates applicable in quantum technologies has been shown previously \cite{Barenco95,bullock-2004-4,Shende09}. While synthesis approaches devised for general unitary matrices would work on Hermitians as well, for the reasons explained in the paper, we believe exploiting the special features of Hermitians would produce superior results. Section 3.5 of the paper shows that the proposed approach can lead to better results than one of the general well-known quantum synthesis approaches, QSD \cite{Shende06}, for synthesizing $CU$ gates and the general synthesis approach for $C^kU$  gates, \cite{Barenco95} where $U$ is a single-qubit Hermitian gate. If the quantum
circuit that needs to be synthesized comprises Hermitian matrices, the proposed approach can be applied to gain better results. Otherwise, the general methods
can still be applied.

To describe our proposed method for quantum-logic synthesis of Hermitian matrices, the remainder of the paper is organized as follows. In Section \ref{sec:BasicCo}, some basic concepts about quantum computation and the Jacobi method are explained. The proposed synthesis approach is introduced in Section \ref{sec:method}. Section \ref{sec:Conc} concludes the paper.

\section{Preliminaries and Definitions}\label{sec:BasicCo}
\subsection{Quantum Basic Concepts}
A quantum bit, named as \emph{qubit}, is a quantum state with two basis states $|0\rangle$ and $|1\rangle$. Based on the superposition principle, a qubit can take any linear combination of its two basis states, i.e., $|\psi\rangle = \alpha|0\rangle+\beta|1\rangle$.
In this equation, $\alpha$ and $\beta$ are complex numbers such that $|\alpha|^2 + |\beta|^2 = 1$.

If the qubit is measured in the computational, i.e., $\{$$\left\vert 0\right\rangle$, $\left\vert 1\right\rangle$$\}$ basis, the classic outcome of 0 is observed with the probability of $|\alpha|^2$ and the classic outcome of 1 is observed with the probability of $|\beta|^2$. If 0 is observed, the state of the qubit after the measurement collapses to $|0\rangle$. Otherwise, it collapses to $|1\rangle$.

A matrix $U$ is \emph{unitary} if $UU^\dag=I$, in which $U^\dag$ is the conjugate transpose of $U$ and $I$ is the identity matrix.
An $n$-qubit quantum gate corresponds to a $2^n\times2^n$ unitary matrix which performs a particular operation on $n$ qubits. Various quantum gates with different functionalities have been introduced. Quantum circuits constructed from a set of quantum gates are often synthesized using either a ``basic gate" library \cite{Barenco95}, with CNOT and single-qubit gates, or an  ``elementary gate" library \cite{bullock-2004-4}, with CNOT and single-qubit rotation gates. Single-qubit rotation gates around $y$ and $z$ axes with the angle of $\theta$ have the matrix representations as illustrated in (\ref{gate1}).
\begin{equation}\label{gate1}
R_y (\theta ) = \left[ {\begin{array}{*{20}c}
{\cos {\textstyle{\theta  \over 2}}} & {\sin {\textstyle{\theta  \over 2}}}  \\
{ - \sin {\textstyle{\theta  \over 2}}} & {\cos {\textstyle{\theta  \over 2}}}  \\
\end{array}} \right]\,\,\,\
R_z (\theta) = \left[ {\begin{array}{*{20}c}
{e{\textstyle{{ - i\theta } \over 2}}} & 0  \\
0 & {e{\textstyle{{i\theta } \over 2}}}  \\
\end{array}} \right]
\end{equation}
\normalsize

There is another set of useful single-qubit gates called Pauli gates (\ref{gate2}). This set together with the identity matrix span the full vector space of two-dimensional matrices.

\begin{equation}\label{gate2}
\sigma _x = X = \left[ \begin{array}{*{20}c}
   0 & 1\\
   1 & 0\\
 \end{array} \right]{\kern 1pt} \,\,\,\sigma _y = Y  =\left[ {\begin{array}{*{20}c}
   0 & { - i}  \\
   i & 0  \\
\end{array}} \right]\,\,\,\sigma _z = Z = \left[ \begin{array}{*{20}c}
   1 & 0  \\
   0 & { - 1}  \\
 \end{array} \right]
\end{equation}

Phase shift gates are a set of single-qubit gates which leave the $\left\vert 0\right\rangle$ state unchanged and map $\left\vert 1\right\rangle$ to $e^{i\alpha}$$\left\vert 1\right\rangle$ as shown in (\ref{gate3}).
Some common phase shift gates are phase gate ($S$), $\frac{\pi}{8}$($T$) and $\sigma _z$ where $\alpha=\frac{\pi}{2},\ \frac{\pi}{4}$ and $\pi$, respectively.
\begin{equation}\label{gate3}
R(\alpha)   = \left[ \begin{array}{*{20}c}
   1 & 0  \\
   0 & {e^{i\alpha } }  \\
 \end{array} \right]
\end{equation}

$C^kU$ gates applying on $n$ qubits for $1\leq k\leq n-1$ are a set of quantum gates with one target and $k$ control qubits. These control qubits can be either positive or negative. If the initial states of the positive control(s) and the negative control(s) are $\left\vert 1\right\rangle$ and $\left\vert 0\right\rangle$ respectively, then $U$ gate is applied to the target qubit and no action is taken otherwise.
CNOT and CZ are two examples of $CU$ gates with a positive control.
They operate on two qubits, i.e., control and target qubits and if the control qubit is $\left\vert 1\right\rangle$, then $\sigma_x$ and $\sigma_z$ gates are performed on the target qubit, respectively, and otherwise the state of the qubit is left unchanged.

In this paper, a $C^{n-1}U$ gate which operates on $n$ qubits is called a multiple-control $U$ gate. Multiple-control $U$ gates can be decomposed to CNOT and single-qubit gates as shown in \cite{Barenco95}.

A two-level unitary matrix~\cite{Nielsen10} is named after a set of unitary matrices which operate non-trivially only on two vector components. Two-level matrices differ with the identity matrix in four elements placed in the indices $pp$, $pq$, $qp$ and $qq$. Multiple-control $U$ gates are some examples of two-level matrices.

If the elements outside the main diagonal of a square matrix are all zero, the matrix is called diagonal. The eigenvalues of a diagonal matrix are its diagonal elements. The matrix representation of CZ gate is an example of a diagonal matrix.

A square matrix can be partitioned into smaller square matrices with the same size, called blocks. If blocks outside the main diagonal are zero matrices, a matrix is called block-diagonal.
A quantum multiplexer~\cite{Shende06} over $n$ qubits, has $m$ target qubits and $s=n-m$ select qubits. A different quantum gate is applied on the targets according to the values of select qubits.
In the case that the select qubits are the most significant ones, quantum multiplexers have a block-diagonal matrix representation with $U_i(2^m)$ matrices, $0 \leq i \leq 2^{n-m}-1$, on main diagonal blocks.
In this case, each $U_i(2^m)$ gate is applied to the target qubit(s) when the select qubit(s) are in the state $\left\vert i\right\rangle$.
If $U_i$'s are $R_z$ (respectively $R_y$), the quantum multiplexer is called multiplexed $R_z$ (respectively $R_y$) gate.
A select qubit in a quantum multiplexer is denoted by $\Box$ as used in~\cite{Shende06}.
If a quantum multiplexer has a single select bit which is the most significant one, it can be written as $U_0\oplus U_1$ where $U_0$ and $U_1$ are applied on the target qubits when the select qubit is $|0\rangle$ and $|1\rangle$ respectively.
\subsection{Hermitian Matrix Properties}
A matrix $A$ is called Hermitian~\cite{Nielsen10} or self-adjoint if $A^\dag=A$. Every Hermitian matrix is normal and therefore, it is diagonalizable as shown in~(\ref{eq0}).
\begin{equation}\label{eq0}
\centering
A=U D U^\dag
\end{equation}
The elements of $D$ are the eigenvalues of $A$ and the columns of the unitary matrix $U$ are the eigenvectors of $A$.
All eigenvalues of a Hermitian matrix are real numbers. On the other hand, the eigenvalues of a unitary matrix have modulus equal to 1 and therefore, the eigenvalues of a Hermitian unitary matrix, the elements of the matrix $D$, are either +1 or -1.
Pauli matrices are some examples of Hermitian matrices.

It should be noted that if a circuit solely consists of Hermitian quantum gates, its matrix is not necessarily Hermitian, since Hermitian matrices are closed under tensor product but they are not closed under matrix multiplication.

In this paper, symmetric, Hermitian and diagonal Hermitian quantum gates which operate on $n$ qubits are denoted by $\mathbb{S}(2^n)$, $\mathbb{H}(2^n)$ and $\mathbb{D}_n$, respectively.

\subsection{Jacobi Method}\label{sec:Jacob}
The Jacobi method includes a set of algorithmic solutions for the real symmetric eigenvalue problem \cite{golub2012matrix}. The main idea of these algorithms is to reduce the norm of the off-diagonal elements shown in (\ref{eq1}), in a systematic manner.

\begin{equation}\label{eq1}
\centering
{\rm off}(A) = \sqrt {\sum\limits_{p = 1}^n {\sum\limits_{\scriptstyle q = 1 \hfill \atop
  \scriptstyle q \ne p \hfill}^n {a_{pq}^2 } } }
\end{equation}

This is done by the help of Jacobi rotations, or Givens rotations denoted by $G$, to make each off-diagonal element zero in each iteration of the Jacobi method. For an $n$$\times$$n$ matrix, $G_{pq}$ is a two-level matrix, i.e., it is similar to an identity matrix $I_n$ except for four elements $pp$, $pq$, $qp$ and $qq$ where $p$ and $q$ are the indices that specify the row and column of the element which is targeted to become zero ($p$$<$$q$). These elements are illustrated in~(\ref{eq2}) as $G(\theta)$.

In the $j^{th}$ Jacobi iteration, the given matrix $A^{(j-1)}$ with a non-zero off-diagonal element at $pq$ is converted to the new matrix $A^{(j)}= G_{pq}^T A^{(j-1)}G_{pq}$ where $G_{pq}^T$ is the transpose of $G_{pq}$. The angle $\theta$ is chosen as shown in (\ref{eq2}) in order to set the $pq$ element to zero.

\begin{equation}\label{eq2}
\centering
G(\theta)  = \left[ {\begin{array}{*{20}c}
   {\cos {\theta \over 2 }} & { \sin {\theta \over 2} }  \\
   { - \sin {\theta \over 2} } & {\cos {\theta \over 2} }  \\
\end{array}} \right]
,\,\,\,\,\,\,\theta  = \arctan (\frac{{ - 2a_{pq} }}{{a_{pp}  - a_{qq} }})
\end{equation}

The final result is a diagonal matrix, and the diagonal elements are the eigenvalues of the $A^{(0)}$ matrix in the first iteration.

\subsection{Jacobi Method for Hermitian Matrices}\label{sec:JacobHerm}
To extend the Jacobi method for Hermitian matrices, Jacobi rotations should be replaced by their complex counterparts. A typical complex rotation, $Q(\theta,\alpha)$ is defined as (\ref{eq4}) which can be used in a similar manner as $G({\theta})$ to construct a Jacobi rotation, where $\alpha$ and $\theta$ can be computed using (\ref{eq7}).

\begin{equation}\label{eq4}
\centering
Q(\theta,\alpha)  = \left[ {\begin{array}{*{20}c}
   {\cos {\theta \over 2} } & {e^{i\alpha } \sin {\theta \over 2} }  \\
   { -e^{-i\alpha } \sin {\theta \over 2} } & {\cos {\theta \over 2} }  \\
\end{array}} \right]
\end{equation}

\begin{equation}\label{eq7}
\tan (\theta ) = \frac{{ - 2\left| {a_{pq} } \right|}}{{a_{pp}  - a_{qq} }},
\,\,\,\,\,\,
e^{i\alpha }  = \frac{{a_{pq} }}{{\left| {a_{pq} } \right|}}
\end{equation}

There is another complex rotation matrix introduced in \cite{park1993real} which is used in our proposed method. This complex rotation matrix, $Q'(\theta,\alpha)$, is defined as (\ref{eq5}), where $G(\theta)$ is a real Jacobi rotation described earlier, and $R(\alpha)$ is the phase shift matrix as introduced in Section~\ref{sec:BasicCo}.

\begin{equation}\label{eq5}
\centering
Q'(\theta,\alpha) = R(-\alpha)G(\theta)  = \left[ {\begin{array}{*{20}c}
   {\cos {\theta \over 2} } & { \sin {\theta \over 2} }  \\
   { -e^{-i\alpha } \sin {\theta \over 2} } & {e^{-i\alpha } \cos {\theta \over 2} }  \\
\end{array}} \right]
\end{equation}

The use of the complex rotation matrix is shown in (\ref{eq6}), where $p$ and $q$ are the indices of the element in the matrix $A$ which should become zeroed $(p<q)$. This definition is used in our synthesis method in Section \ref{sec:method}.

\begin{equation}\label{eq6}
\centering
\begin{array}{l}
 A^{(j)}  = {Q'}_{pq}^{\dag} A^{(j - 1)} {Q'}_{pq}  \\\\
  = \left[ {\begin{array}{*{20}c}
   {c } & { - s }  \\
   {s} & {c}  \\
\end{array}} \right]\left[ {\begin{array}{*{20}c}
   1 & 0  \\
   0 & {e^{i\alpha } }  \\
\end{array}} \right]\left[ {\begin{array}{*{20}c}
   {a_{pp}^{(j-1)} } & {a_{pq}^{(j-1)} }  \\
   {a_{qp}^{(j-1)} } & {a_{qq}^{(j-1)} }  \\
\end{array}} \right]\left[ {\begin{array}{*{20}c}
   1 & 0  \\
   0 & {e^{ - i\alpha } }  \\
\end{array}} \right]\left[ {\begin{array}{*{20}c}
   {c} & {s }  \\
   { - s} & {c}  \\
\end{array}} \right], \\\\
 c = {\cos {\theta \over 2} }, \,\,
 s = {\sin {\theta \over 2} }
 \end{array}
\end{equation}

The parameters $\theta$ and $\alpha$ can be calculated using (\ref{eq7}) where $-\frac{\pi}{2} \leq \theta \leq \frac{\pi}{2}$ \cite{park1993real}.
In the remainder of this paper, whenever the indices $p$ and $q$ are not important, $G_{pq}$ and $R_{pq}$ are simply referred to as $G$ and $R$.
\section{Jacobi-Based Quantum-Logic Synthesis}\label{sec:method}
In this section, the main structure of the proposed Jacobi-based method for Hermitian quantum gate synthesis (JBHS) is discussed.
 The Jacobi method described in Sections \ref{sec:Jacob} and \ref{sec:JacobHerm} is used in our synthesis algorithm as a matrix decomposition method. It transforms a given Hermitian unitary matrix $\mathbb{H}(2^n)$ to a set of two-level matrices, their conjugate transposes and a single diagonal matrix as shown in (\ref{eq8}) by a recursive function. The elements of $\mathbb{H}^{(j)}$ are denoted by $h^{(j)}_{pq}$ for the $j^{th}$ iteration.
\begin{equation}\label{eq8}
\begin{array}{l}
 \left\{ \begin{array}{l}
 \mathbb{H}(2^n)^{(0)}  = \mathbb{H}(2^n) \\
 \mathbb{H}(2^n)^{(j)}  = G_{pq}^{\dag} R_{pq}^{\dag} \mathbb{H}^{(j - 1)}(2^n) R_{pq} G_{pq} \,{\kern 1pt} \,\, \\
 \end{array} \right. \\\\
 {\kern 1pt} 1 \le j \le 2^{2n-1}  - 2^{n-1},\,\,\,\,
 0 \le p<q \le 2^n-1 ,\,\,\,\,h^{(j-1)}_{pq} \ne 0\\
 \end{array}
\end{equation}
Solving the recursive function of (\ref{eq8}) leads to (\ref{eq8-1}) where $j_{\max }$ shows the total number of iterations which is equal to the number of non-zero elements in $\mathbb{H}^{(j)}$ matrix whose row index is less than column index.

\begin{equation}\label{eq8-1}
\begin{array}{l}
 \mathbb{D}_n = G^{\dag (j)} R^{\dag (j)} ...G^{\dag (1)} R^{\dag (1)} \mathbb{H}(2^n)R^{(1)} G^{(1)} ...R^{(j)} G^{(j)}  \\
 \mathbb{D}_n =\prod\limits_{j = j_{\max } }^1 {\{ G^{\dag (j)} R^{\dag (j)} \} } \mathbb{H}(2^n)\prod\limits_{j = 1}^{j_{\max } } {\{ R^{(j)} G^{(j)} \} }  \\
 \end{array}
\end{equation}

Based on (\ref{eq8-1}), the general structure of the proposed JBHS is shown in (\ref{eq8-2}).
\begin{equation}\label{eq8-2}
\begin{array}{l}
 \mathbb{H}(2^n) = R^{(1)} G^{(1)} ...R^{(j)} G^{(j)} \mathbb{D}_{n}G^{\dag (j)} R^{\dag (j)} ...G^{\dag (1)} R^{\dag (1)}  \\
 \mathbb{H}(2^n) = \prod\limits_{j = 1}^{j_{\max } } {\{ R^{(j)} G^{(j)} \} } \mathbb{D}_{n}\prod\limits_{j = j_{\max } }^1 {\{ G^{\dag (j)} R^{\dag (j)} \} }  \\
 \end{array}
\end{equation}

In the remainder of this section, the quantum equivalents of the produced two-level matrices and the middle diagonal matrix are discussed independently and then the whole structure of the proposed synthesis method in (\ref{eq8-2}) is presented by synthesizing a general $\mathbb{H}(4)$ matrix as an example.
\subsection{Quantum-Equivalence of the Two-Level Matrices}
In this section, first, possible multiple-control $U$ gates on $n$ qubits and their corresponding matrices are introduced. Then, in Theorem \ref{theo1}, two-level matrices produced by the proposed synthesis method are described as quantum operators according to the definition of multiple-control gates.

An $n$-qubit multiple-control $U$ gate is denoted as $C^{n-1}U_i^j$, where $0\leq i\leq n-1$ and $0\leq j\leq 2^{n-1}-1$. The target qubit of $C^{n-1}U_i^j$ is the $i^{th}$ qubit and the ($n$-1)-bit binary expression of $j$ represents the control string, $0$ for negative and $1$ for positive control(s).
$C^{n-1}U_i^j$ gates are two-level matrices which change only two basis states, i.e., $|a\rangle$ and $|b\rangle$ where the binary expressions of $a$ and $b$ differ only in one bit, i.e., they are two adjacent gray codes. $a$ and $b$ can be computed from $i$ and $j$ by an injective function as follows.
The binary expressions of $a$ (respectively $b$) are obtained by inserting a zero (respectively one) in the $i^{th}$ bit of the binary expression of $j$.
The matrix of $C^{n-1}U_i^j$ is the same as an $I^{ \otimes n}$ matrix except for the four elements of $U$ which are placed in the indices $aa$, $ab$, $ba$ and $bb$.

It should be noted that the number of different $C^{n-1}U_i^j$ gates is equal to $n2^{n-1}$. Since there are $\binom{n}{1}=n$ possible target qubits for these gates and for each target qubit, $2^{n-1}$ different control strings can be assumed. As an example, for $n=2$, Figure \ref{fig1} shows four different $C U_i^j$ gates corresponding to the matrices in (\ref{eq9}).

\begin{equation}\label{eq9}
\begin{array}{l}
\begin{array}{*{20}c}
   {\,\,\,\,\,\,\,\,\,\,\,\,\,00} & {01} & {10} & {11}  \\
\end{array}\begin{array}{*{20}c}
   {\,\,\,\,\,\,\,\,\,\,\,\,\,\,\,\,\,\,\,\,\,\,\,\,00} & {01} & {10} & {11}  \\
\end{array} \\
 \begin{array}{*{20}c}
   {00}  \\
   {01}  \\
   {10}  \\
   {11}  \\
\end{array}\left[ {\begin{array}{*{20}c}
   {u_{00} } & 0 & {u_{01} } & 0  \\
   0 & 1 & 0 & 0  \\
   {u_{10} } & 0 & {u_{11} } & 0  \\
   0 & 0 & 0 & 1  \\
\end{array}} \right]{\kern 1pt} {\kern 1pt} \,\,\,\,\,\,\,\begin{array}{*{20}c}
   {00}  \\
   {01}  \\
   {10}  \\
   {11}  \\
\end{array}\left[ {\begin{array}{*{20}c}
   1 & 0 & 0 & 0  \\
   0 & {u_{00} } & 0 & {u_{01} }  \\
   0 & 0 & 1 & 0  \\
   0 & {u_{10} } & 0 & {u_{11} }  \\
\end{array}} \right] \\\\
\,\,\,\,\,\,\,\,\,\,\,\,\,\,i = 0\,,j = 0{\kern 1pt} \,\,\,\,\,\,\,\,\,\,\,\,\,\,\,\,\,\,\,\,\,\,\,\,\,\,\,\,i = 0\,,j = 1 \\\\
 \begin{array}{*{20}c}
   {\,\,\,\,\,\,\,\,\,\,\,\,\,00} & {01} & {10} & {11}  \\
\end{array}\begin{array}{*{20}c}
   {\,\,\,\,\,\,\,\,\,\,\,\,\,\,\,\,\,\,\,\,\,\,\,\,00} & {01} & {10} & {11}  \\
\end{array} \\
 \begin{array}{*{20}c}
   {00}  \\
   {01}  \\
   {10}  \\
   {11}  \\
\end{array}\left[ {\begin{array}{*{20}c}
   {u_{00} } & {u_{01} } & 0 & 0  \\
   {u_{10} } & {u_{11} } & 0 & 0  \\
   0 & 0 & 1 & 0  \\
   0 & 0 & 0 & 1  \\
\end{array}} \right]{\kern 1pt} {\kern 1pt} \,\,\,\,\,\,\,\begin{array}{*{20}c}
   {00}  \\
   {01}  \\
   {10}  \\
   {11}  \\
\end{array}\left[ {\begin{array}{*{20}c}
   1 & 0 & 0 & 0  \\
   0 & 1 & 0 & 0  \\
   0 & 0 & {u_{00} } & {u_{01} }  \\
   0 & 0 & {u_{10} } & {u_{11} }  \\
\end{array}} \right] \\\\
\,\,\,\,\,\,\,\,\,\,\,\,\,\,i = 1\,,j = 0{\kern 1pt} \,\,\,\,\,\,\,\,\,\,\,\,\,\,\,\,\,\,\,\,\,\,\,\,\,\,\,\,i = 1\,,j = 1 \\\\
 \end{array}
\end{equation}

\begin{figure}
\centering
\small
\[
\Qcircuit @C=0.7em @R=0.5em @!R{
			&\gate{U}\qwx[1]	&\gate{U}\qwx[1]	&\ctrlo{1}			&\ctrl{1}&\qw		&\\
			&\ctrlo{0}			&\ctrl{0}		     &\gate{U}			&\gate{U}	&\qw		&
}
\]
\caption{\label{fig1}Corresponding gates of the matrices in (\ref{eq9}). From left to right: ($i=0$ and $j=0$), ($i=0$ and $j=1$), ($i=1$ and $j=0$) and ($i=1$ and $j=1$).}
\end{figure}
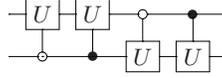
\normalsize
\begin{theorem}\label{theo1}
Each two-level matrix produced in the Jacobi iteration can be decomposed into an $n$-qubit multiple-control $R_y(\theta)$, an $n$-qubit multiple-control $R(\alpha)$ gate and a set of multiple-control NOT gates which act on $n$-qubits.
\end{theorem}

 \begin{proof}If the binary expressions of $p$ and $q$ in the $pq$ index of $G_{pq}$ and $R_{pq}$ in (\ref{eq8}) are two adjacent gray codes, there will be an $n$-qubit multiple-control $R_y(\theta)$, i.e., $C^{n-1}R_y(\theta)_i^j$, equivalent to $G_{pq}$, and an $n$-qubit multiple-control $R(\alpha)$, i.e., $C^{n-1}R(\alpha)_i^j$, equivalent to $R_{pq}$, where $i$ and $j$ can be easily obtained from $p$ and $q$, as explained earlier.

Now consider the case that the difference between the binary expressions of $p$ and $q$ is more than one bit, i.e., $l$ bits where $2\leq l \leq n$.
There exists a gray code sequence that connects the binary expressions of $p$ with $q$ using $l+1$ elements (containing the source and destination numbers). $l-2$ $n$-qubit multiple-control NOT gates are required to map the binary expression of $p$ to an adjacent gray code of $q$, denoted by $g$. Then an $n$-qubit multiple-control $R(\alpha)$, i.e., $C^{n-1}R(\alpha)_i^j$, and an $n$-qubit multiple-control $R_y(\theta)$, i.e., $C^{n-1}R_y(\theta)_i^j$ are applied. $i$ and $j$ can be easily computed from $p$ and $g$ as mentioned before. Finally, $l-2$ multiple-control NOT gates are required to put  $|g\rangle$ back to $|p\rangle$.
It is worth mentioning that if $h_{pq}$ in (\ref{eq8}) is a real number, no multiple-control $R(\alpha)$ gate is required.
\end{proof}
\subsection{Quantum-Equivalence of the Diagonal Matrix} \label{sec:diag}
The produced diagonal matrix $\mathbb{D}_n$ in (\ref{eq7}) can be synthesized either by previous general synthesis methods for quantum diagonal matrices such as \cite{bullock-2004-4} or by specific methods for Hermitian diagonal matrices as the one presented in~\cite{diagonal}. Figure \ref{fig2} shows the synthesis of an arbitrary diagonal matrix $\Delta_{3}$ for $n=3$ qubits. The synthesis method of~\cite{diagonal} was presented using the fact that the diagonal elements of Hermitian gates are either +1 or -1 and hence they can be synthesized using a set that solely consists of multiple-control Pauli \emph{Z} gates. The authors of~\cite{diagonal} introduced a binary representation for the diagonal Hermitian gates and showed that the binary representations of multiple-controlled \emph{Z} gates form a basis for the vector space that is produced by the binary representations of all diagonal Hermitian quantum gates. Finally, the problem of decomposing a given diagonal Hermitian gate was mapped to the problem of writing its binary
representation in the specific basis mentioned above. It was shown that this approach can lead to circuits with lower costs in comparison with the approach of~\cite{bullock-2004-4}.

\begin{figure}[t]
\centering
\small
\[
\Qcircuit @C=0.7em @R=0.5em @!R{
&\multigate{2}{\Delta_3}    &\qw &   &&&\qw		&\qw			&\gate{R_z}		        &\qw\\
&\ghost{\Delta_3}           &\qw & = &&&\qw		&\gate{R_z}		&\gate{}\qwx[-1]		&\qw		\\
&\ghost{\Delta_3}           &\qw &   &&&\gate{R_z}	&\gate{}\qwx[-1]	&\gate{}\qwx[-1]		&\qw		
}
\]
\caption{\label{fig2}Synthesis of $\Delta_3$ gate \cite{bullock-2004-4}.}
\end{figure}
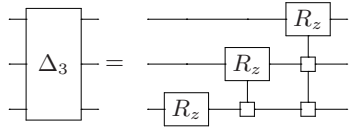

\subsection{Gate-Order Analysis of the Proposed Method}

To find the number of produced gates by the proposed JBHS method after applying the pure synthesis method to synthesize a given $\mathbb{H}(2^n)$ quantum gate without any possible optimizations, four kinds of gates are considered:
$C^{n-1}R_y(\theta)$, $C^{n-1}R(\alpha)$ and $C^{n-1}$NOT gates besides a diagonal Hermitian gate which can be synthesized by the method of \cite{bullock-2004-4}, as mentioned in Section \ref{sec:diag}, by at most $2^n-2$ CNOT gates.

For a general $\mathbb{H}(2^n)$ matrix, there are at most $2^{2n-1}-2^{n-1}$ non-zero elements which should become zeroed. Therefore, in the worst-case, $2^{2n}-2^n$ gates of each $C^{n-1}R_y(\theta)$ and $C^{n-1}R(\alpha)$ types, the latter only for complex elements, are needed.
Among these $2^{2n-1}-2^{n-1}$ non-zero elements, $h_{ij}$: $0$$\leq$$ i$$<$$j$ $\leq$ $2^n-1$, there are $n2^{n-1}$ elements where the binary representations of $i$ and $j$ differ in only one bit and therefore, no $C^{n-1}$NOT gate is needed for them. The other $2^{2n-1}-n2^{n-1}-2^{n-1}$ elements need at most $4(n-2)C^{n-1}$NOT gates. As a result, the numbers of produced $C^{n-1}R_y(\theta)$ and $C^{n-1}R(\alpha)$ gates are of the order $O(4^n)$ and the number of needed $C^{n-1}$NOT gates is of the order $O(n4^n)$. It is worth mentioning that zero elements eliminate their related gates without any effects on the other non-zero elements. For sparse $\mathbb{H}(2^n)$ matrices, the actual number of needed gates is much less, since the number of gates is a coefficient of non-zero elements. Moreover, possible post-synthesis optimizations decrease the number of required elementary gates. These post-synthesis optimizations include eliminating the control lines of the two produced Hermitian conjugate multiple-control $U$ gates around a middle multiple-control $Z$ gate. An example of this optimization is shown in Fig.~\ref{fig:multiple}.
\subsection{Possible Options in the Synthesis Procedure}
There are two major options for the order of selecting $p$ and $q$ in (\ref{eq8}) during the synthesis process which are discussed in the following propositions.
Using these propositions which are based on the inherent parallelism of the Jacobi algorithm, elements can be selected in order to produce circuits with better results in terms of the number of elementary gates.

\textbf{Proposition 1.} Off-diagonal elements of a given Hermitian gate on $n$ qubits, $h_{pq}$ $p<q$, can be divided into $2^n-1$ independent sets, namely the computations of these sets have no conflicts.

\begin{proof}According to the definition of \emph{parallel-ordering} problem~\cite{golub2012matrix}, $(p_1,q_1)$$(p_2,q_2)$$,...,$$(p_N,q_N)$, $N=(2^n-1)2^{n-1}$, is a parallel ordering of a set $\{(p,q)|0\leq p<q\leq 2^n-1\}$ if for $s$ from $0$ to $2^n-2$, the rotation set $rotation.set(s)=\{(p_k,q_k):k=2^{n-1}  s+1: 2^{n-1} (s+1)\}$ consists of rotations with no conflicts which results into $2^n-1$ independent sets. Therefore, $(2^n-1)!$ different arrangements of these sets are possible during or after the JBHS synthesis process.
Any arrangement of these sets determines the arrangements of their conjugate transposes on the other side of the middle diagonal gate.
\end{proof}

\textbf{Proposition 2.} The resulted gates of each member in every set of Proposition 1 are interchangeable with the resulted gates of other members in that set.

\begin{proof}Based on Proposition 1, each set consists of non-conflicting members. Therefore, their computations have no conflicts and make their resulted gates interchangeable. It should be noted that any arrangement of these gates of each member determines the arrangements of their conjugate transposes on the other side of the middle diagonal gate.
\end{proof}

These options provide the possibility of arranging these parallel sets and the gates inside each set in the produced circuit to cancel some redundant gates.

A Hermitian quantum gate $\mathbb{H}(4)$ is considered in (\ref{eq10}). Since the matrix is Hermitian, setting the off-diagonal elements $h_{pq}$, $p<q$, to zero by (\ref{eq8-2}) turns the other off-diagonal elements into zero too.

\begin{equation}\label{eq10}
\begin{array}{l}
 \,\,\,\,\,\,\,\,\,\,\,\,\,\,\,00\,\,\,\,\,01\,\,\,\,\,10\,\,\,\,\,11 \\
 \begin{array}{*{20}c}
   {00}  \\
   {01}  \\
   {10}  \\
   {11}  \\
\end{array}\left[ {\begin{array}{*{20}c}
   {d_{00} } & {h_{01} } & {h_{02} } & {h_{03} }  \\
   {} & {d_{11} } & {h_{12} } & {h_{13} }  \\
   {} & {} & {d_{22} } & {h_{23} }  \\
   {} & {} & {} & {d_{33} }  \\
\end{array}} \right] \\
 \end{array}
\end{equation}

According to Proposition 1, there are three independent sets \{(0,2)(1,3)\}, \{(0,1)(2,3)\}, \{(0,3)(1,2)\}. The synthesized circuit according to these sets are illustrated in Figure \ref{fig3}. The order of gates is the same as the order of the sets from left to right. The independent sets which can be exchanged with each other at the left side of the middle diagonal gate are separated by dash lines.
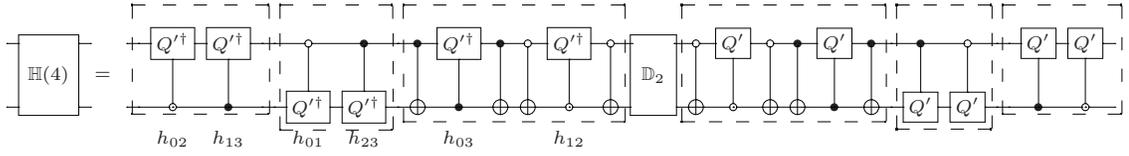
\begin{figure}[t!]
\centering
\scriptsize
\[
\Qcircuit @C=0.65em @R=0.0000001em @!R{
&		        					&    										& 	&&&    	&				&				&    	&	&				&			 &	&		&				&		&		&				&		&		         	&		&				&		&		&				&		&	&				 &				&	&    &				&				&		\\
&\multigate{2}{\mathbb{H}(4)} \gategroup{1}{7}{4}{10}{.7em}{--}  &\qw \gategroup{1}{12}{4}{13}{.7em}{--}\gategroup{1}{15}{4}{20}{.7em}{--}	&   	&&&\qw 	 &\gate{Q'^{\dag}}\qwx[2]	&\gate{Q'^{\dag}}\qwx[2]	&\qw 	&\qw	&\ctrlo{2}			&\ctrl{2}		&\qw	&\ctrl{2}	&\gate{Q'^{\dag}}\qwx[2]	 &\ctrl{2}	&\ctrlo{2}	&\gate{Q'^{\dag}}\qwx[2]	&\ctrlo{2}	&\multigate{2}{\mathbb{D}_2} 	&\ctrlo{2}	&\gate{Q{'}}\qwx[2]		&\ctrlo{2}	&\ctrl{2}	 &\gate{Q{'}}\qwx[2]		&\ctrl{2}	&\qw	&\ctrl{2}			&\ctrlo{2}			&\qw	&\qw &\gate{Q{'}}\qwx[2]	&\gate{Q{'}}\qwx[2]		&\qw		\\
&		        					&    										& = 	&&&    	&				&				&    	&	&				&			 &	&		&				&		&		&				&		&		         	&		&				&		&		&				&		&	&				 &				&	&    &				&				&		\\
&\ghost{\mathbb{H}(4)}  \gategroup{1}{22}{4}{27}{.7em}{--}	&\qw \gategroup{1}{29}{4}{30}{.7em}{--}	\gategroup{1}{32}{4}{35}{.7em}{--}	&   	&&&\qw 	 &\ctrlo{0}			&\ctrl{0}			&\qw 	&\qw	&\gate{Q'^{\dag}}		&\gate{Q'^{\dag}}	&\qw	&\targ		&\ctrl{0}			&\targ		&\targ		 &\ctrlo{0}			&\targ		&\ghost{\Delta_2}        	&\targ		&\ctrlo{0}			&\targ		&\targ		&\ctrl{0}			&\targ		&\qw	 &\gate{Q{'}}			&\gate{Q{'}}			&\qw	&\qw &\ctrl{0}			&\ctrlo{0}			&\qw		\\
&								&    										&	&&&    	&h_{02}				&h_{13}				&    	&	&h_{01}				 &h_{23}			&	&		&h_{03}				&		&		&h_{12}				&		&			 	&		&				&		&		&				 &		&	&				&				&	&    &				&				&				
}
\]
\caption{\label{fig3}Synthesis of a Hermitian quantum gate for $n$$=$$2$ qubits.}
\end{figure}
General decomposition of a symmetric gate on two qubits is illustrated in Figure \ref{fig3.1}.
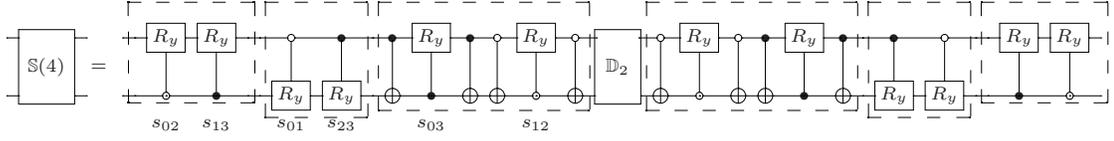
\begin{figure}[t!]
\centering
\scriptsize
\[
\Qcircuit @C=0.65em @R=0.0000001em @!R{
&		        					&    										& 	&&&    	&			&			&    	&	&			&		&	&		 &			&		&		&			&		&		         	&		&			&		&		&			&		&	&				&				 &	&    &			&			&		\\
&\multigate{2}{\mathbb{S}(4)} \gategroup{1}{7}{4}{10}{.7em}{--}  &\qw \gategroup{1}{12}{4}{13}{.7em}{--}	\gategroup{1}{15}{4}{20}{.7em}{--}	&   	&&&\qw 	 &\gate{R_y}\qwx[2]	&\gate{R_y}\qwx[2]	&\qw 	&\qw	&\ctrlo{2}		&\ctrl{2}	&\qw	&\ctrl{2}	&\gate{R_y}\qwx[2]	&\ctrl{2}	&\ctrlo{2}	 &\gate{R_y}\qwx[2]	&\ctrlo{2}	&\multigate{2}{\mathbb{D}_2} 	&\ctrlo{2}	&\gate{R_y}\qwx[2]	&\ctrlo{2}	&\ctrl{2}	&\gate{R_y}\qwx[2]	&\ctrl{2}	&\qw	 &\ctrl{2}			&\ctrlo{2}			&\qw	&\qw &\gate{R_y}\qwx[2]	&\gate{R_y}\qwx[2]	&\qw		\\
&		        					&    										& = 	&&&    	&			&			&    	&	&			&		&	&		 &			&		&		&			&		&		         	&		&			&		&		&			&		&	&				&				 &	&    &			&			&		\\
&\ghost{\mathbb{S}(4)}  \gategroup{1}{22}{4}{27}{.7em}{--}	&\qw \gategroup{1}{29}{4}{30}{.7em}{--}	\gategroup{1}{32}{4}{35}{.7em}{--}	&   	&&&\qw 	 &\ctrlo{0}		&\ctrl{0}		&\qw 	&\qw	&\gate{R_y}		&\gate{R_y}	&\qw	&\targ		&\ctrl{0}		&\targ		&\targ		&\ctrlo{0}		&\targ		 &\ghost{\Delta_2}        	&\targ		&\ctrlo{0}		&\targ		&\targ		&\ctrl{0}		&\targ		&\qw	&\gate{R_y}			&\gate{R_y}			&\qw	 &\qw &\ctrl{0}		&\ctrlo{0}		&\qw		\\
&								&    										&	&&&    	&s_{02}			&s_{13}			&    	&	&s_{01}			&s_{23}		 &	&		&s_{03}			&		&		&s_{12}			&		&			 	&		&			&		&		&			&		&	&				 &				&	&    &			&			&				
}
\]
\caption{\label{fig3.1}Synthesis of a symmetric quantum gate for $n$$=$$2$ qubits.}
\end{figure}
Using a quantum multiplexer notation from \cite{Shende06} and the optimization rules of \cite{arabzadeh2010rule} the circuit of Figure \ref{fig4} is obtained. Additional optimizations lead to the circuit of Figure \ref{fig4.1} by eliminating more CNOT gates.

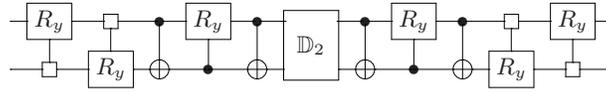
\begin{figure}[t!]
\centering
\small
\[
\Qcircuit @C=0.7em @R=0.5em @!R{
	&\gate{R_y}\qwx[1]	&\gate{}\qwx[1]	&\ctrl{1}	&\gate{R_y}\qwx[1]	&\ctrl{1}	&\multigate{1}{\mathbb{D}_2}	&\ctrl{1}	&\gate{R_y}\qwx[1]	&\ctrl{1}	 &\gate{}\qwx[1]	&\gate{R_y}\qwx[1]			&\qw		\\
	&\gate{}		&\gate{R_y}	&\targ		&\ctrl{0}		&\targ		&\ghost{\Delta_2} 		&\targ		&\ctrl{0}		&\targ		 &\gate{R_y}	&\gate{}				 &\qw		
}
\]
\caption{\label{fig4}The circuit of Figure \ref{fig3.1} after applying some optimizations from \cite{arabzadeh2010rule} on CNOT gates and merging controlled-$R_y$ gates to reach a multiplexed $R_y$ with fewer CNOT and single-qubit rotation gates.}
\end{figure}

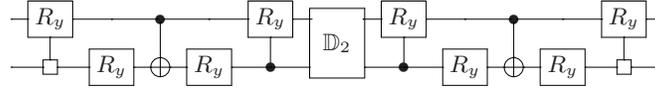
\begin{figure}[t!]
\centering
\small
\[
\Qcircuit @C=0.7em @R=0.5em @!R{
	&\gate{R_y}\qwx[1]	&\qw		&\ctrl{1}	&\qw		&\gate{R_y}\qwx[1]	&\multigate{1}{\mathbb{D}_2}	&\gate{R_y}\qwx[1]	&\qw		&\ctrl{1}	&\qw		 &\gate{R_y}\qwx[1]			&\qw		\\
	&\gate{}		&\gate{R_y}	&\targ		&\gate{R_y}	&\ctrl{0}		&\ghost{\Delta_2} 		&\ctrl{0}		&\gate{R_y}	&\targ		&\gate{R_y}	&\gate{}				 &\qw		
}
\]
\caption{\label{fig4.1}The circuit of Figure \ref{fig4} after merging the adjacent CNOT gates into $\mathbb{D}_2$, which produces a new $\mathbb{D}_2$ gate. Two CNOT gates, next to multiplexed $R_y$ gates, can be canceled out by one CNOT in the decomposition of each multiplexed $R_y$ in accordance to \cite{Shende06}.}
\end{figure}

\subsection{Application: Synthesis of Multiple-Control Hermitian Gates}
The Jacobi method for calculating the eigenvalues and eigenvectors of Hermitian matrices was implemented in  MATLAB.
In this section, the results of applying the JBHS approach on a special set of circuits,  C$\mathbb{H}(2)$ gates and their general case with $k$ control qubits, are considered.

It can be readily verified that an $\mathbb{H}(2)$ quantum gate (except $I$ and $-I$) can be written as $\mathbb{H(\theta,\alpha)}$:
\begin{equation}
\label{eq:h2}
\mathbb{H(\theta,\alpha)}=
\left[ \begin{array}{l}
 \cos (\theta )\,\,\,\,\,e^{ - i\alpha } \sin (\theta ) \\
 e^{i\alpha } \sin (\theta )\,\,\,-\cos (\theta ) \\
 \end{array} \right],
\end{equation}
where $\alpha$ is a real parameter and $0\leq\theta\leq\pi$. Applying the JBHS method on C$\mathbb{H(\theta,\alpha)}$ will lead to the following decomposition:
\begin{equation}
\label{eq:jac}
\mathrm{C}\mathbb{H(\theta,\alpha)}=\mathrm{CR(\alpha)}\mathrm{CR_y(-\theta)}\mathrm{CZ}\mathrm{CR_y(\theta)}\mathrm{CR(-\alpha)}.
\end{equation}
The controls of the synthesized gates of the off-diagonal parts can be eliminated as they are their Hermitian conjugates. Using this optimization,~(\ref{eq:jac}) can be written as:
\begin{equation}
\label{eq:jac1}
\mathrm{C}\mathbb{H(\theta,\alpha)}=(I \otimes A)\mathrm{CZ} (I \otimes B),
\end{equation}
where $A=R(\alpha)R_y(-\theta)$ and $B=R_y(\theta)R(-\alpha)$.

The C$\mathbb{H(\theta,\alpha)}$ decomposition using the JBHS method is shown in Figure~\ref{fig7}.
\begin{figure}[!tb]
\
\small
\[
\Qcircuit @C=0.8em @R=0.0000001em @!R{
&&\qw	&\ctrl{2} 					&\qw  \gategroup{3}{16}{3}{17}{.7em}{--}		& &\qw	&\ctrl{2}				&\ctrl{2}			&\ctrl{2}	&\ctrl{2}			 &\ctrl{2}						 &\qw	&	&\qw		&\qw				 &\qw				&\ctrl{2}	&\qw				&\qw					&\qw		 \\
&&	&         					& 							&=&	&			        	&				&   		&		    		&			      	        		 &	&=	&		 &				 &				&   		&		    		&			      	 	&		\\
&&\qw	&\gate{\mathbb{H(\theta,\alpha)}}    		&\qw  \gategroup{3}{19}{3}{20}{.7em}{--}		& &\qw	&\gate{R(-\alpha)}			&\gate{R_y(\theta)}		 &\ctrl{0}	&\gate{R_y(-\theta)}		&\gate{R(\alpha)}					&\qw	&	&\qw		&\gate{R(-\alpha)}		&\gate{R_y(\theta)}		&\ctrl{0}	 &\gate{R_y(-\theta)}	&\gate{R(\alpha)}				&\qw		\\
&&	&										&							& &	&					&				&		&				&							 &	&	&		&				&				&		&			&						&																																				 
}
\]
\caption{\label{fig7}The synthesized circuit of C$\mathbb{H(\theta,\alpha)}$ using JBHS method.}
\end{figure}
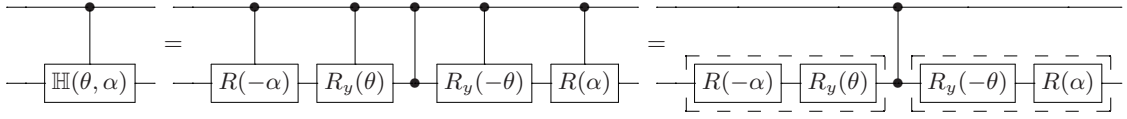

{Using the approach presented in~\cite[Lemma 5.5]{Barenco95}, the following decomposition is obtained for C$\mathbb{H(\theta,\alpha)}$ gates,}
\begin{equation}
\mathrm{C}\mathbb{H(\theta,\alpha)}=(I \otimes P) \mathrm{CNOT} (I \otimes Q),
\end{equation}
where $P=R_z(\alpha)R_y(-\theta+\frac{\pi }{2})$ and $Q=R_y(\theta-\frac{\pi }{2})R_z(-\alpha)$.
The C$\mathbb{H(\theta,\alpha)}$ decomposition using~\cite[Lemma 5.5]{Barenco95} is shown in Figure~\ref{fig8}.
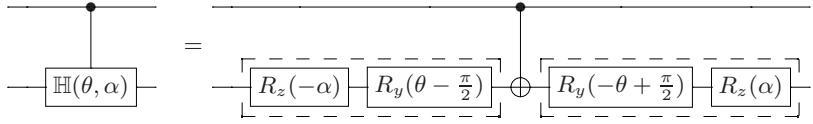
\begin{figure}[!tb]
\
\small
\[
\Qcircuit @C=0.8em @R=0.0000001em @!R{
&	&\qw	&\ctrl{2} 				\gategroup{3}{10}{3}{11}{.7em}{--}	&\qw	&		& &&\qw	&\qw					&\qw					&\ctrl{2}	&\qw					 &\qw							&\qw		\\
&	&	&         									&	&		&=&&	&			        	&					&   		&		    			&			      	        		 &	      	\\
&	&\qw	&\gate{\mathbb{H(\theta,\alpha)}}    	\gategroup{3}{13}{3}{14}{.7em}{--}	&\qw	&		& &&\qw	&\gate{R_z(-\alpha)}			 &\gate{R_y(\theta-\frac{\pi}{2})}	&\targ{0}	&\gate{R_y(-\theta+\frac{\pi}{2})}	&\gate{R_z(\alpha)}					&\qw		\\
&	&	&										&	&		& &&	&					&					&		&					&							 &																																				
}
\]
\caption{\label{fig8}The synthesized circuit of C$\mathbb{H(\theta,\alpha)}$ using the method of~\cite[Lemma 5.5]{Barenco95}.}
\end{figure}

Decomposition of these gates using the QSD method~\cite{Shende06} is also calculated. In QSD, each C$\mathbb{H(\theta,\alpha)}$ gate is considered as a single-select qubit quantum multiplexer and is synthesized as follows:
\begin{equation}
\label{eq:qsd}
\mathrm{C}\mathbb{H(\theta,\alpha)}=(I \otimes V)(D\oplus D^\dagger) (I \otimes W),
\end{equation}
where $V=R(\alpha)R_y(-\theta)$, $W=SR_y(\theta)R(-\alpha)\mathbb{H}(\theta,\alpha)$ and $D=S$. The middle diagonal gate in~(\ref{eq:qsd}) is indeed a multiplexed $R_z$ gate whose target qubit is the first one.
Applying the QSD to synthesize $\mathrm{C\mathbb{H}(\theta,\alpha)}$ gates produces a circuit structure as shown in Figure~\ref{fig9}.

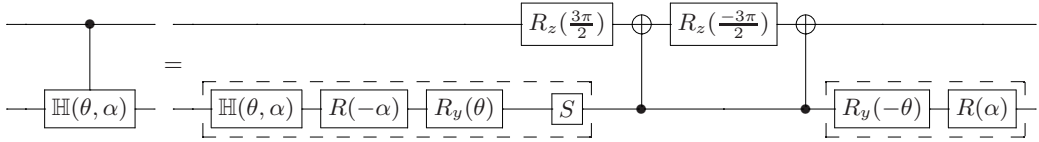
\begin{figure}[!tb]
\
\small
\[
\Qcircuit @C=0.8em @R=0.0000001em @!R{
& &\qw	\gategroup{3}{8}{3}{11}{.7em}{--}	&\ctrl{2} 					&\qw	&	&\qw	&\qw					&\qw			&\qw			 &\gate{R_z(\frac{3\pi}{2})}	&\targ{0}	&\gate{R_z(\frac{-3\pi}{2})}	&\targ{0}	&\qw			&\qw			&\qw	\\
& &						&         					&	&=	&	&					&			&			&				&   		&				&   		 &			&			&	\\
& &\qw	\gategroup{3}{15}{3}{16}{.7em}{--}	&\gate{\mathbb{H(\theta,\alpha)}}    		&\qw	&	&\qw	&\gate{\mathbb{H(\theta,\alpha)}} 	 &\gate{R(-\alpha)}	&\gate{R_y(\theta)}	&\gate{S}			&\ctrl{-2}	&\qw				&\ctrl{-2}	&\gate{R_y(-\theta)}	&\gate{R(\alpha)}	&\qw	\\
& &						&						&	&	&	&					&			&			&				&		&				&		&			 &			&																																			
}
\]
\caption{\label{fig9}The synthesized circuit of C$\mathbb{H(\theta,\alpha)}$ using the method of~\cite{Shende06}, QSD.}
\end{figure}

Table \ref{tab:res} shows the obtained decompositions to synthesize $\mathrm{C\mathbb{H}(\theta,\alpha)}$ gates using the proposed JBHS,~\cite[Lemma 5.5]{Barenco95} and the QSD methods.
Each CZ gate can be implemented using a CNOT gate at the cost of inserting two single-qubit rotation gates around $y$ axis as shown in Figure~\ref{fig:CZCNOT}.
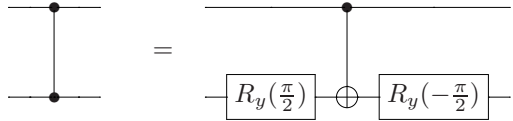
\begin{figure}[!tb]
\
 \[
\Qcircuit @C=0.8em @R=0.0000001em @!R{
&&&&\qw	&\ctrl{2}	&\qw   	&  \qw							&&&&&		&\qw				 &\ctrl{2}	&\qw		 			&\qw		\\
&	&&&	&       	&  &	&			 			&&=	&&	&&&		&        			 &          	 &         				&		 Â\\
&	&&&\qw	&\ctrl{0} 	&\qw &\qw 	&  		&&	& &	&\gate{R_y(\frac{\pi}{2})}	&\targ	&\gate{R_y(-\frac{\pi}{2})}	&\qw		
}
\]
\caption{\label{fig:CZCNOT}Circuit equivalence of CZ and CNOT gates.}
\end{figure}

Table \ref{tab:res2} compares the number of produced gates.
Although the proposed JBHS method and the method of~\cite{Barenco95} produce the same number of elementary gates, the JBHS approach directly synthesizes $\mathrm{C\mathbb{H}(\theta,\alpha)}$ gates to a library that consists of CZ and single-qubit rotation gates around $y$ and $z$ axes. The CZ gate is of interest as it is supported as a primitive operation by four quantum physical machine descriptions (PMD) while CNOT gate is supported by only two PMDs ~\cite{ftqls2014}.
CZ gates are also useful in producing a parallel structure for quantum circuits~\cite{extended} using one-way quantum computation model~\cite{par}, as the input quantum circuits to that procedure are assumed to contain CZ gates.


 \renewcommand{\arraystretch}{1.3}

\begin{table}[t]
\tbl{Synthesis comparison of $\mathrm{C\mathbb{H}(\theta,\alpha)}$ gates. $\mathrm{CNOT}^{2,1}$ denotes a CNOT gate with the control on the second and target on the first qubit.\label{tab:res}}{

\small
\begin{tabular}{|l|l|l|}
\hline	
Gate 			&Method 				&Synthesized circuit																 \\
\hline
\hline
\multirow{3}{*}{CH}	&JBHS					&$(I\otimes R_y(\frac{-\pi}{4}))$CZ$(I\otimes R_y(\frac{\pi}{4}))$										\\* \cline{2-3}

			&\cite[Lemma 5.5]{Barenco95}		&$(I\otimes R_y(\frac{\pi}{4}))$CNOT$(I\otimes R_y(\frac{-\pi}{4}))$										 \\* \cline{2-3}

			&QSD \cite{Shende06}			&$(I\otimes R_y(\frac{-\pi}{4}))$$\mathrm{CNOT}^{2,1}$$(R_z(\frac{-3\pi}{2})\otimes I)$$\mathrm{CNOT}^{2,1}$$(R_z(\frac{3\pi}{2})\otimes(SR_y(\frac{\pi}{4})H))$		\\* \cline{2-3}
\hline
\hline
\multirow{3}{*}{CY} 	&JBHS					&$(I\otimes SR_y(\frac{-\pi}{2}))$CZ$(I\otimes R_y(\frac{\pi}{2})S^{\dag})$									\\* \cline{2-3}

			&\cite[Lemma 5.5]{Barenco95}		&$(I\otimes R_z(\frac{\pi}{2}))$CNOT$(I\otimes R_z(\frac{-\pi}{2}))$										 \\* \cline{2-3}
			
			&QSD \cite{Shende06}			&$(I\otimes SR_y(\frac{-\pi}{2}))$$\mathrm{CNOT}^{2,1}$$(R_z(\frac{-3\pi}{2})\otimes I)$$\mathrm{CNOT}^{2,1}$$(R_z(\frac{3\pi}{2})\otimes(SR_y(\frac{\pi}{2})S^{\dag}Y))$	\\* \cline{2-3}
\hline
\hline
\multirow{3}{*}{CNOT} 	&JBHS					&$(I\otimes R_y(\frac{-\pi}{2}))$CZ$(I\otimes R_y(\frac{\pi}{2}))$										 \\* \cline{2-3}

			&\cite[Lemma 5.5]{Barenco95}		&CNOT																		 \\* \cline{2-3}

			&QSD \cite{Shende06}			&$(I\otimes R_y(\frac{\pi}{2}))$$\mathrm{CNOT}^{2,1}$$(R_z(\frac{-3\pi}{2})\otimes I)$$\mathrm{CNOT}^{2,1}$$(R_z(\frac{3\pi}{2})\otimes(SR_y(\frac{\pi}{2})X))$		\\* \cline{2-3}
\hline
\hline
\multirow{3}{*}{CZ} 	&JBHS					&CZ																		 \\* \cline{2-3}

			&\cite[Lemma 5.5]{Barenco95}		&$(I\otimes R_y(\frac{\pi}{2}))$CNOT$(I\otimes R_y(\frac{-\pi}{2}))$										 \\* \cline{2-3}
			
			&QSD \cite{Shende06}			&$\mathrm{CNOT}^{2,1}$$(R_z(\frac{-3\pi}{2})\otimes I)$$\mathrm{CNOT}^{2,1}$$(R_z(\frac{3\pi}{2})\otimes (SZ))$								 \\* \cline{2-3}

\hline
\end{tabular}}
\end{table}

%

\renewcommand{\arraystretch}{1.3}

\begin{table}[t]
\tbl{Comparison of the number of produced CNOT and single-qubit rotation gates around $y$ and $z$ axis and CZ and single-qubit rotation gates around $y$ and $z$ axis after the synthesis of $\mathrm{C\mathbb{H}(\theta,\alpha)}$ gates. The JBHS method directly produces CZ gates and the method of \cite{Barenco95} and \cite{Shende06} directly produce CNOT gates.\label{tab:res2}}{

\small
\begin{tabular}{|l|l|c|c||c|c|}
\hline	
Gate 			&Method 				&\#(CNOT)&\#($R_y$,$R_z$)	&\#(CZ)	&\#($R_y$,$R_z$)	 \\
\hline
\hline
\multirow{3}{*}{CH}	&JBHS					&1	&2		&1	&2			\\* \cline{2-6}

			&\cite[Lemma 5.5]{Barenco95}		&1	&2		&1	&2			 \\* \cline{2-6}

			&QSD \cite{Shende06}			&2	&6		&2	&10			\\* \cline{2-6}
\hline
\hline
\multirow{3}{*}{CY} 	&JBHS					&1	 &2		&1	 &4			\\* \cline{2-6}

			&\cite[Lemma 5.5]{Barenco95}		&1	&2		&1	&4			 \\* \cline{2-6}
			
			&QSD \cite{Shende06}			&2	&7		&2	&11			\\* \cline{2-6}
\hline
\hline	
\multirow{3}{*}{CNOT} 	&JBHS					&1	&0		&1	&2			 \\* \cline{2-6}

			&\cite[Lemma 5.5]{Barenco95}		&1	&0		&1	&2			\\* \cline{2-6}

			&QSD \cite{Shende06}			&2	&6		&2	&10			\\* \cline{2-6}
\hline
\hline
\multirow{3}{*}{CZ} 	&JBHS					&1	&2		&1	&0			\\* \cline{2-6}

			&\cite[Lemma 5.5]{Barenco95}		&1	&2		&1	&0			 \\* \cline{2-6}
			
			&QSD \cite{Shende06}			 &2	&4		 &2	&8			\\* \cline{2-6}

\hline
\end{tabular}}
\end{table}

%

Proposition 3 shows how applying the JBHS method on multiple-control $\mathbb{H}(2)$ gates can lead to an implementation which requires a linear number of elementary gates in terms of circuit lines.

\textbf{Proposition 3.} Using one auxiliary qubit with an arbitrary state, any multiple-control $\mathbb{H}(2)$ gate on $n$-qubits can be decomposed to $O(n)$ elementary gates, using the proposed JBHS method.

\begin{proof} If $\mathbb{H}(2)$ is $I$ or $-I$ gate, then the JBHS method will produce $I^{\otimes {n}}$ or $-I^{\otimes {n}}$ gates which require no elementary gates to be implemented. Otherwise, each specific $\mathbb{H}(2)$ quantum gate can be written as $\mathbb{H}(\theta,\alpha)$ using~(\ref{eq:h2}). Applying the JBHS method to multiple-control $\mathbb{H}(2)$ gates will lead to a circuit structure similar to Figure~\ref{fig:multiple}.
The middle multiple-control $Z$ gate can be decomposed to a multiple-control NOT gate at the cost of inserting two rotation gates around $y$ axis (Figure~\ref{fig:CZCNOT}). This can in turn be decomposed to $O(n)$ elementary gates using one auxiliary qubit with an arbitrary state by the approach presented in~\cite{maslov2008}.
\end{proof}

\renewcommand{\arraystretch}{1.3}

\begin{table}[t]
\tbl{Comparison of the number of produced CZ and single-qubit (1-qu) gates for decomposing $C^{n-2}U$ gates where $U$ is a single-qubit Hermitian gate.\label{tab:res3}}{

\scriptsize
\begin{tabular}{|l|c|c|c|c|c|c|c|c|}
\hline	
\multirow{3}{*}{Method} 	&\multicolumn{8}{c|}{Number of qubits}				 \\
				&\multicolumn{2}{c|}{7}	&\multicolumn{2}{c|}{8}		&\multicolumn{2}{c|}{9}		&\multicolumn{2}{c|}{$n$}	\\* \cline{2-9}
				&\#CZ	 &\#1-qu	&\#CZ	 	&\#1-qu	&\#CZ	 	&\#1-qu	&\#CZ	 	&\#1-qu	\\
\hline				
JBHS				&\textbf{84}	&\textbf{98}		&\textbf{108}		&\textbf{122}		&\textbf{168}		&\textbf{146}		&\textbf{24$n$-48}		 &\textbf{24$n$-70}		\\
\hline
\cite[Collary 7.12]{Barenco95}	&122		&124			&170			&172			&218			&220			&48$n$-214			&48$n$-212		 \\
\hline
\end{tabular}}
\end{table}

\begin{figure}[!tb]
\centering
\scriptsize
\[
\Qcircuit @C=0.8em @R=0.0000001em @!R{
&n-1	&&&/\qw	&\ctrl{2} 					&\qw  \gategroup{3}{26}{3}{27}{.7em}{--}		& &&n-1	&&&/\qw	&\ctrl{2}				&\ctrl{2}			 &\ctrl{2}	&\ctrl{2}			 &\ctrl{2}						 &\qw	&	&&n-1	&&&/\qw		&\qw				 &\qw				 &\ctrl{2}	&\qw				 &\qw					&\qw		\\
&	&&&	&         					& 							&=&&	&&&	&			        	 &				&   		&		    		&			      	        		 &	&=	&&	 &&&		 &				 &				&   		&		    		&			      	 	&		 \\
&	&&&\qw	&\gate{\mathbb{H(\theta,\alpha)}}    		&\qw  \gategroup{3}{29}{3}{30}{.7em}{--}		 & &&	&&&\qw	&\gate{R(-\alpha)}			 &\gate{R_y(\theta)}		&\ctrl{0}	 &\gate{R_y(-\theta)}		&\gate{R(\alpha)}					&\qw	&	&&	&&&\qw		 &\gate{R(-\alpha)}		 &\gate{R_y(\theta)}		&\ctrl{0}	&\gate{R_y(-\theta)}	 &\gate{R(\alpha)}				&\qw		\\
&	&&&	&						&							& &&	&&&	&					&				 &		&				&							 &	&	&&	&&&		&				&				 &		&			&						&																														 
}
\]
\caption{\label{fig:multiple}The synthesized circuit of multiple-control $\mathbb{H}(\theta,\alpha)$ gate on $n$ qubits.}
\end{figure}
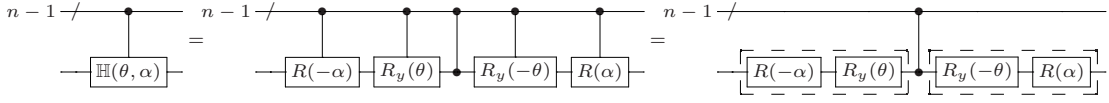

It should be noted that an arbitrary multiple-control $U$ gate can also be implemented using linear number of elementary gates, using one auxiliary qubit by~\cite[Collary 7.12]{Barenco95}. However, the auxiliary qubit should be initially fixed in the state of $|0\rangle$. Auxiliary qubits with an arbitrary state, in contrast to qubits with fixed states, can be employed in the rest of the circuit for other computations. Besides, Table \ref{tab:res3} is provided to compare the gate counts produced by the proposed JBHS approach and the approach of~\cite[Collary 7.12]{Barenco95}. The synthesized circuits resulted from the two approaches, one with an auxiliary qubit in arbitrary state and the other with the auxiliary qubit fixed to $|0\rangle$, are considered. To do this, the results of~\cite[Collary 7.4]{Barenco95} and \cite{Shende09} are used to decompose the produced multiple-control NOT gates to
CZ and single-qubit gates.
As shown in the table, the number of both CZ and single-qubit gates improves the results of ~\cite[Collary 7.12]{Barenco95}.

As some examples, decomposition results of controlled-$Y$ and controlled-Hadamard gates are shown in (\ref{eq-cn}) and (\ref{eq-ch}), respectively.
The synthesized circuits of multiple-control $Y$ and multiple-control Hadamard gates on $n$ qubits are illustrated in Figures \ref{fig5} and \ref{fig6}, respectively.

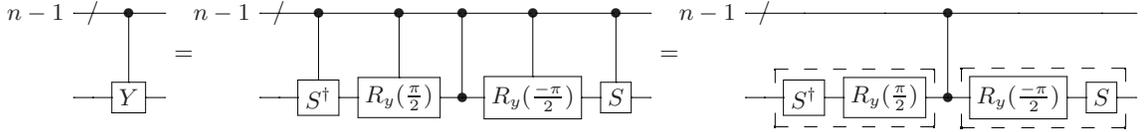
\begin{figure}[!tb]
\centering
\small
\[
\Qcircuit @C=0.8em @R=0.0000001em @!R{
&n-1	&&&/\qw	&\ctrl{2} 		&\qw  	& &&n-1	&&&/\qw	&\ctrl{2}		 \gategroup{3}{26}{3}{27}{.7em}{--}		&\ctrl{2}			&\ctrl{2}	&\ctrl{2}			 &\ctrl{2}					 &\qw	&	&&n-1	&&&/\qw		&\qw				 &\qw				&\ctrl{2}	&\qw				&\qw					&\qw		 \\
&	&&&	&         		& 	&=&&	&&&	&							        	&				&   		&		    		&			      	        	 &	&=	&&	&&&		 &				 &				&   		&		    		&			      	 	&		\\
&	&&&\qw	&\gate{Y}    		&\qw  	& &&	&&&\qw	&\gate{S^{\dag}}	 \gategroup{3}{29}{3}{30}{.7em}{--}		&\gate{R_y(\frac{{\pi }}{2})}	&\ctrl{0}	 &\gate{R_y(\frac{{-\pi }}{2})}	&\gate{S}					&\qw	&	&&	&&&\qw		&\gate{S^{\dag}}		&\gate{R_y(\frac{{\pi }}{2})}	&\ctrl{0}	 &\gate{R_y(\frac{{-\pi }}{2})}	&\gate{S}				&\qw		
}															
\]
\caption{\label{fig5}The synthesized circuit of multiple-control $Y$ gate on $n$ qubits.}
\end{figure}

\begin{figure}[!tb]
\centering
\small
\[
\Qcircuit @C=0.8em @R=0.0000001em @!R{
&n-1	&&&/\qw	&\ctrl{2}	&\qw   	&  	&&n-1	&&&/\qw	&\ctrl{2}				&\ctrl{2}		&\ctrl{2}					&\qw	&&	&&n-1	&&&/\qw		&\qw				 &\ctrl{2}	&\qw		 			&\qw		\\
&	&&&	&       	&  	&=  	&&	&&&	&			 		&   		    	&						&	&&=	&&	&&&		&        			&          	&         				 &		 \\
&	&&&\qw	&\gate{H} 	&\qw  	&  	&&	&&&\qw	&\gate{R_y (\frac{{\pi }}{4})}		&\ctrl{0}		&\gate{R_y (\frac{{-\pi }}{4})}			&\qw	&&	&&	 &&&\qw		&\gate{R_y (\frac{{\pi }}{4})}	&\ctrl{0}	&\gate{R_y (\frac{{-\pi }}{4})}		&\qw		
}
\]
\caption{\label{fig6}The synthesized circuit of multiple-control Hadamard gate on $n$ qubits.}
\end{figure}
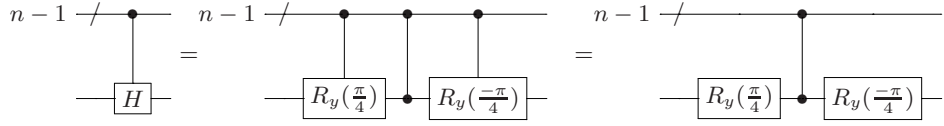

\begin{equation}\label{eq-cn}
\small
\begin{array}{l}
 \left[ {\begin{array}{*{20}c}
   1 & 0 & 0 & 0  \\
   0 & 1 & 0 & 0  \\
   0 & 0 & 0 & -i  \\
   0 & 0 & i & 0  \\
\end{array}} \right] =
 \left[ {\begin{array}{*{20}c}
      1 & 0 & 0 & 0  \\
   0 & 1 & 0 & 0  \\
   0 & 0 & 1 & 0  \\
   0 & 0 & 0 & i  \\
\end{array}} \right]\left[ {\begin{array}{*{20}c}
1 & 0 & 0 & 0  \\
   0 & 1 & 0 & 0  \\
   0 & 0 & {0.7071} & {-0.7071}  \\
   0 & 0 & {0.7071} & {0.7071}  \\
\end{array}} \right] \left[ {\begin{array}{*{20}c}
   1 & 0 & 0 & 0  \\
   0 & 1 & 0 & 0  \\
   0 & 0 & 1 & 0  \\
   0 & 0 & 0 & { - 1}  \\
\end{array}} \right]
\left[ {\begin{array}{*{20}c}
    1 & 0 & 0 & 0  \\
   0 & 1 & 0 & 0  \\
   0 & 0 & {0.7071} & {0.7071}  \\
   0 & 0 & {-0.7071} & {0.7071}  \\
\end{array}} \right] \left[ {\begin{array}{*{20}c}
 1 & 0 & 0 & 0  \\
   0 & 1 & 0 & 0  \\
   0 & 0 & 1 & 0  \\
   0 & 0 & 0 & -i  \\
\end{array}} \right] \\
 \end{array}
\end{equation}
\begin{equation}\label{eq-ch}
\small
\begin{array}{l}
 \left[ {\begin{array}{*{20}c}
   1 & 0 & 0 & 0  \\
   0 & 1 & 0 & 0  \\
   0 & 0 & {2^{ - 0.5} } & {2^{ - 0.5} }  \\
   0 & 0 & {2^{ - 0.5} } & { - 2^{ - 0.5} }  \\
\end{array}} \right] =
 \left[ {\begin{array}{*{20}c}
   1 & 0 & 0 & 0  \\
   0 & 1 & 0 & 0  \\
   0 & 0 & { 0.9239} & {-0.3827}  \\
   0 & 0 & { 0.3827} & { 0.9239}  \\
\end{array}} \right]\left[ {\begin{array}{*{20}c}
   1 & 0 & 0 & 0  \\
   0 & 1 & 0 & 0  \\
   0 & 0 & 1 & 0  \\
   0 & 0 & 0 & { - 1}  \\
\end{array}} \right]\left[ {\begin{array}{*{20}c}
   1 & 0 & 0 & 0  \\
   0 & 1 & 0 & 0  \\
   0 & 0 & { 0.9239} & { 0.3827}  \\
   0 & 0 & {-0.3827} & { 0.9239}  \\
\end{array}} \right] \\
 \end{array}
\end{equation}



\section{Conclusions and Future Works}\label{sec:Conc}
The problem of quantum-logic synthesis of Hermitian quantum gates was addressed in this paper. The Jacobi-based synthesis approach, JBHS was introduced that uses the Jacobi method to decompose a given matrix to a set of two-level matrices and a middle diagonal Hermitian matrix. The quantum-gate equivalence of this matrix decomposition was discussed and the structure of the circuit and its possible optimizations were described.

Finally, the results of applying the JBHS method on multiple-control $\mathbb{H}(2)$ gates were presented to demonstrate how the proposed method can synthesize these gates using a linear number of elementary gates in terms of circuit lines, with the aid of one auxiliary qubit in an arbitrary state.

Some further improvements can be applied to the proposed approach. As a future work, finding the best order for zeroing non-diagonal elements during the JBHS method in order to reduce the number of produced elementary gates is being considered.





\end{document}